\documentclass{dmgt}

\usepackage{amsmath,amssymb}
\usepackage{url}
\usepackage{graphicx}
\usepackage{enumerate}
\usepackage{enumitem}
\usepackage{braket}
\usepackage{amsfonts}
\usepackage{amsthm}
\usepackage{fancyhdr}
\usepackage{comment}
\usepackage{float}
\usepackage{graphicx}
\usepackage{mathtools}
\usepackage{xspace}
\usepackage{pgf,tikz,pgfplots}
\pgfplotsset{compat=1.14}
\usepackage{mathrsfs}
\usetikzlibrary{arrows}
\usepackage{lscape}

\newauthor{Liliana Alc\'{o}n}{L. Alc\'{o}n}{%
 Universidad Nacional de La Plata, La Plata, Argentina.\\CONICET}[%
liliana@mate.unlp.edu.ar]

\newauthor{%
Mar\'{i}a P\'{i}a Mazzoleni}{%
M. P. Mazzoleni}{%
Universidad Nacional de La Plata, La Plata, Argentina.}[%
pia@mate.unlp.edu.ar]

\newauthor{%
Tanilson Dias dos Santos}{%
T. D. Santos}{%
Federal University of Tocantins, Palmas, Brazil}[%
tanilson.dias@mail.uft.edu.br]

\title{ On $B_1$-EPG  and EPT graphs}

\classnbr{05C62 - Graph representations}

\begin{document}
%\nolinenumbers
\begin{abstract}
This research contains  
as a main result the prove that every Chordal $B_1$-EPG graph is simultaneously in the graph classes VPT and EPT. In addition, 
we  describe structures that must be  present in any $B_1$-EPG  graph which does  not admit a Helly-$B_1$-EPG representation. 
In particular, this paper presents  some features of non-trivial families of graphs properly contained in Helly-$B_1$ EPG, namely Bipartite, Block, Cactus and Line of Bipartite  graphs. 
\end{abstract}

\section{Introduction}

Models based on paths intersection  may consider  intersections by vertices or   intersections by edges.  Cases where the paths are hosted on a tree  appear first in the literature, see for instance \cite{gavril1978recognition, golumbic1985edge, golumbic1985}.  Representations using paths on a grid were considered later, see  \cite{golumbic2009,golumbic2013, golumbic2013intersection}. %More details on each intersection model will be given in the following text.

 Let $P$ be a family of paths on a host tree $T$ . Two types of intersection graphs from the pair $<P,T>$ are defined, namely VPT and EPT graphs.
The \textit{edge intersection graph} of $P$, EPT(P), has vertices which correspond to the members of $P$, and two vertices are adjacent in EPT(P) if and only if the corresponding paths in $P$ share at least one edge in T. Similarly, the \textit{vertex intersection graph} of $P$, VPT(P), has vertices which correspond to the members of $P$, and two vertices are adjacent in VPT(P) if and only if the corresponding paths in $P$ share at least one vertex in $T$.
VPT and EPT graphs are incomparable families of graphs. However, when the maximum degree of the host tree is restricted to three the family of
VPT graphs coincides with the family of EPT graphs \cite{golumbic1985edge% \cite{alcon2010necessary
}. Also it is known that any Chordal EPT graph is VPT (see~\cite{syslo1985triangulated}). Recall that it was shown that Chordal graphs are the vertex intersection graphs of subtrees of a tree \cite{gavril1974intersection}.

Edge intersection graphs of paths on a grid are called \textit{EPG graphs}. 

In \cite{golumbic2009}, the authors proved that every graph is EPG, and started the study of the subclasses
defined by bounding the number of times any path used in the representation can bend.  Graphs admitting a representation
where  paths  have at most $k$ changes of direction  (bends) were called $B_k$-EPG. 
 In particular, when the paths have at most one bend we have the \textit{ $B_1$-EPG graphs} or a \textit{single bend EPG graphs}.

 A pertinent question in the context of path intersection graphs is as follows: given two classes of path intersection graphs,
 the first whose host is a tree and the second whose host is a grid,  is there an intersection or containment relationship among these classes? What do we know about it?

In the present paper we will explore $B_1$-EPG graphs, in particular diamond-free graphs and Chordal graphs. We will work on the question about the containment
relation between  VPT, EPT and $B_1$-EPG graph classes.

 A collection  of sets satisfies the \textit{Helly property} when every pair-wise intersecting sub-collection  has at least one common element. When this property
 is satisfied by the set of vertices (edges) of the paths used in a representation, we get a Helly representation.  Helly-$B_1$-EPG graphs were studied
 in \cite{dmtcs:6506}.                                     
It is known that not every $B_1$-EPG graph admits a Helly-$B_1$-EPG representation. We are interested in determining the subgraphs that make
$B_1$-EPG graphs do not admit a Helly representation. In the present work, we describe some structures that will be present in any such subgraph,
and, in addition, we present new  Helly-$B_1$ EPG  subclasses.
Moreover,  we  describe new  Helly-$B_1$ EPG  subclasses % that have Helly property 
and we give some sets of subgraphs that delimit Helly subfamilies.   

\section{Definitions and Technical Results}

%\tanilson{nao tenho certeza se o termo vizinhanca fechada ainda eh utilizado no texto... vou verificar isso na proxima leitura}

The \textit{vertex set} and the \textit{edge set} of a graph $G$ are denoted by $V(G)$ and $E(G)$, respectively.  Given a vertex $v\in V(G)$,  $N(v)$  represents the \textit{open
 neighborhood} of $v$ in $G$. 
For a subset $S \subseteq V(G)$,  $G[S]$ is the subgraph of $G$ induced by $S$.
 If $\mathcal{F}$ is any family of graphs, we say that  $G$ is  \textit{$\mathcal{F}$-free} if $G$ has no induced subgraph isomorphic to a member of $\mathcal{F}$.
 A \textit{cycle},  denoted by $C_n$,  is a sequence of distinct
vertices $v_1, \dots , v_n, v_1$  where $v_i \neq v_j$ for $i \neq j$ and $(v_i, v_i + 1) \in E(G)$, such that
$n \geq 3$. A \textit{chord} is an edge that is between two non-consecutive vertices in a sequence of vertices of a cycle. An \textit{induced cycle}  or \textit{chordless cycles} is a cycle that has no chord, in this paper an induce cycle will simply be called  \textit{cycle}. A graph $G$ formed by an induced cycle $H$ plus  a single universal vertex $v$ connected to all vertices of $H$
is called \textit{wheel graph}. If the wheel has $n$ vertices, it is denoted by $n$-wheel. 

The $k$\textit{-sun graph }$S_k$, $k \geq 3$, consists of
$2k$ vertices, an independent set $X = \{x_1, \dots, x_k\}$ and a clique $Y = \{y_1, \dots, y_k\}$, and edges set $E_1 \cup E_2$, where $E_ 1=\{ (x_1,y_1); (y_1, x_2); (x_2, y_2); (y_2, x_3); \dots , (x_k, y_k); (y_k, x_1) \}$ forms the outer cycle and $E_2= \{(y_i, y_j) |i\neq j\}$ forms the inner clique.

A graph is a $ B_k$-EPG graph if it admits an EPG representation in which each path has at most $k$ bends.  When $ k = 1 $ we say that this is a \emph{single bend EPG} representation or simply a $B_1$-EPG representation.
A \textit{clique} is a set of pairwise adjacent vertices and
an \textit{independent set} is a set of pairwise non adjacent vertices.
Given an EPG representation of a graph $G$, we will identify each vertex $v$ of $G$ with the corresponding path $P_{v}$ of the grid used in the representation. Accordingly, for instance, we will say that a vertex of $G$ covers or contains some edge of the grid (meaning that the corresponding path does), or that a set of paths of the representation
induces a subgraph of $G$ (meaning that the corresponding set of vertices does). 

In  a $B_1$-EPG representation, a clique $K$  is said to be
 an \textit{edge-clique} if all  the vertices of $K$ share a common edge of the grid (see Figure~\ref{fig:cliquesRepresentation}(a)).
 A \textit{claw of the grid} is a set of three edges of the grid incident into a same point of the grid, which is called
  the \textit{center of the claw}. The two edges of the claw that have the same direction form
    the \textit{ base of the claw}. If $K$ is not an edge-clique, then there exists
    a claw of the grid (and only one) such that  the vertices of $K$ are those containing exactly two of the three edges of the claw; such a  clique is called  \textit{claw-clique} \cite{golumbic2009} (see Figure~\ref{fig:cliquesRepresentation}(b)).

\begin{figure}[htb]
  \centering
  \begin{tabular}{  p{4cm} p{0.7cm} p{4cm} }
    %\centering
    \includegraphics[width=4.5cm]{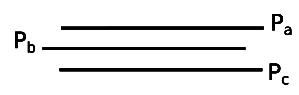} & &
    \includegraphics[width=3.5cm]{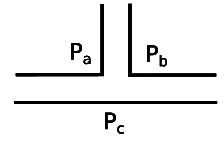}
    \\
    \footnotesize %\centering 
    (a)  \footnotesize Representation of a clique as edge-clique. && \footnotesize (b) Representation  of a clique as claw-clique.\\
  \end{tabular}

 \caption{Examples of clique representations.} \label{fig:cliquesRepresentation}
\end{figure}

Notice that if three vertices induce a claw-clique, then exactly two of them turn at the center of the corresponding  claw of the grid, and the third one contains the
base of the claw. 
Furthermore, any other vertex  adjacent to the three  must contain two of the edges of that claw, then the following lemma holds.

\begin{lemma}\label{lem:cliquesMaximais}
If three vertices are together  in more than one maximal clique of a graph $G$, then in
any $B_1$-EPG representation of $G$ the three vertices do not form a claw-clique. %corresponding paths do not form a claw.
\end{lemma}

%\input{includes/include-img/lemaClaw2Maximais.tex}
 
% \begin{figure}[htb]
%   \centering
%   \begin{tabular}{  p{5cm} p{0.7cm} p{5cm} }
%     %\centering
%     \includegraphics[width=3.5cm]{lemaClaw2Maximais} & &
%     \includegraphics[width=5.5cm]{claw2}
%     \\
%     \footnotesize %\centering 
%     (a)  \footnotesize Example of two maximal cliques sharing vertices. && \footnotesize (b) Representation  of a claw-clique in grid.\\
%   \end{tabular}

%  \caption{Vertices represented by a claw are present in a unique maximal clique} \label{fig:lemaClaw2Maximais}
% \end{figure}
% ESTA FIGURA HAY QUE SACARLA, NO SE LLAMA EN EL TEXTO Y NO SUMA NADA

In \cite{ries2009} Asinowski et al. proved the following lemma for $C_4$-free graphs.

\begin{lemma} \cite{ries2009} \label{lem:lemaBRies2009}
Let $G$ be a $B_1$-EPG graph. If $G$ is $C_4$-free, then there exists a $B_1$-EPG representation of $G$ such that every  maximal claw-clique $K$ is represented on a claw of the grid whose base is covered only by vertices of $K$.
\end{lemma}

We have obtained the following similar result for diamond-free graphs. A \textit{diamond} is a graph $G$ with vertex set $V(G) = \{a, b, c, d\}$ and edge set $E(G)=\{ab, ac,bc, bd,cd\}$.% (QUITAR ESTO Y LA FIGURA 3 see Figure~\ref{fig:diamond}). %A graph is diamond-free if it does not contain a diamond as induced subgraph.

%\input{includes/include-img/diamond.tex}
%  \begin{figure}[htb]	
%  \center%6.3
%  \includegraphics[width=2.2cm]{diamond.png}
%  \caption{Diamond graph.}
% \label{fig:diamond}
% \end{figure}  

\begin{lemma}\label{lem:b1epgDiamondFree}
Let $G$ be a $B_1$-EPG graph. If $G$ is diamond-free, then in any $B_1$-EPG representation of $G$,  every maximal claw-clique $K$ is represented on a claw of the grid whose edges are covered only by vertices of $K$.
\end{lemma}

\begin{proof}Let $K$ be a maximal clique which is a claw-clique in a given $B_1$-EPG representation of $G$. Then there exist three vertices of $K$ which induce a claw-clique $K'$ on
the same claw of the grid than $K$. Assume, in order to derive a contradiction, that a vertex $v\notin K$ covers some edge of the claw. Clearly, $v$ must  cover
only one of such edges. Therefore $v$ and the vertices of $K'$ induce a diamond, a contradiction. 
\end{proof}

% \begin{defi} \label{defi:tortasFrame}

Let $ Q $ be a grid and let $ (a_1, b),$ $(a_2, b),$ $(a_3, b),$ $(a_4, b)$ be a $4$-star centered at $b$ as depicted in Figure~\ref{fig:piesInGrid}(a). Let $ \mathcal{P} = \{P_1, \dots , P_4\}$ be a collection of four paths each containing a different pair of edges of the $4$-star.
%exactly two edges of the $4$-star:
Following \cite{golumbic2009}, we say that the four paths form
\begin{itemize}
\item a \emph{true pie} %is a representation where each $P_i$ of $ \mathcal{P} $ has a bend at $b$.
when each one has a bend at $b$, Figure~\ref{fig:piesInGrid}(b); and 
\item a \emph {false pie} when exactly two of the paths %$P_i$  do not 
bend at $b$ and they do not share an edge of the $4$-star, Figure~\ref{fig:piesInGrid}(c). %contain bends, while the remaining two do not share an edge. 

\begin{figure}[htb]
  \centering
%segundo bloco de figuras
  \begin{tabular}{c c c c c }
    \includegraphics[width=3.5cm]{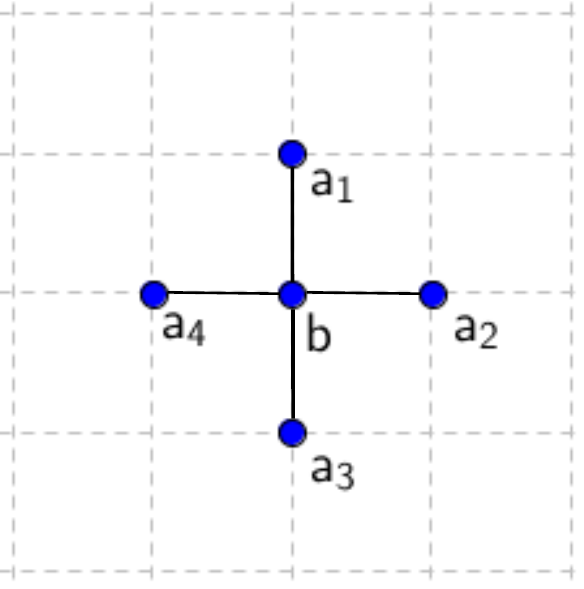}    
    & &\includegraphics[width=3.5cm]{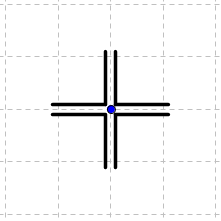} 
    & &
 \includegraphics[width=3.5cm]{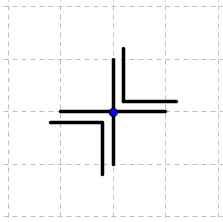} \\%[\abovecaptionskip]
    {\footnotesize (a) 4-star in grid.}  & &  {\footnotesize (b) True pie.} & & {\footnotesize (c) False pie.} %\label{fig:frame}
  \end{tabular}
  \caption{$B_{1}$-EPG representation of the induced cycle of size 4 as pies with emphasis in center $b$.}\label{fig:piesInGrid}
\end{figure}

%\vspace{-0.5cm}
\end{itemize}
% \end{defi}

Clearly if four paths of a $B_1$-EPG representation of $G$ form a pie, then the corresponding vertices induce a $4$-cycle in $G$. % The converse implication is also true (see~\cite{golumbic2009}). 
The following result can be easily proved. We say that a set of paths form a claw when each pair of edges of the claw is covered by some of the paths.

\begin{lemma}\label{lem:twoClawNotSameCenterInChordal}
In any $B_1$-EPG representation of a graph $G$, a set of paths forming two different claws centered at the same point of the grid contains four paths forming either a true pie or a false pie. Therefore, in any $B_1$-EPG representation of a chordal graph $G$, no two maximal claw-cliques of $G$ are centered at the same point of the grid.
\end{lemma}

\begin{lemma}\label{lem:3cliquesNotClaw}
Let $G$ be a graph whose vertex set  can be
partitioned into a non trivial clique $K$ and an independent set $I=\{w_1,w_2,w_3\}$, such that each vertex of $K$ is adjacent to each vertex of $I$. Then, in any $B_1$-EPG representation of $G$, at least one of the cliques  $K_i = K \cup \{w_i\}$, with $1 \leq i \leq 3$,  is an edge-clique. 
\end{lemma}

\begin{proof}
Assume, in order to derive a contradiction, that the three cliques are claw-cliques. By Lemma~\ref{lem:twoClawNotSameCenterInChordal}, they have different centers, say the points $q_1, q_2, q_3$ of the grid, respectively. Since at least two paths have a bend at the center of a claw, for each $i\in\{1,2,3\}$,   there must exist a vertex
  $v_i$ of $K$ such that the corresponding path $P_{v_i}$ turns at the point $q_i$ of the grid.  Notice that each one of the three paths $P_{v_i}$
  must contain  the three grid points $q_1$, $q_2$ and $q_3$. To prove that this is not possible, we will consider, without loss of generality, two cases.
  First,  $q_1$ is between $q_2$ and $q_3$ in $P_{v_1}$. Then, $P_{v_3}$ cannot turn at $q_3$ and contain $q_1$ and $q_2$.   And second,
  $q_2$ is between $q_1$ and $q_3$ in $P_{v_1}$. In this case, $P_{v_2}$ cannot turn at $q_2$ and contain $q_1$ and $q_3$; thus the proof is completed.
 \hfill 
\end{proof}

Three vertices $u, v, w$ of a graph $G$ form an \textit{asteroidal triple} (AT) of $G$ if for every pair of them there exists a path connecting the two vertices and such that the path avoids the neighborhood of the remaining vertex~\cite{Asinowski2009}. A graph without an asteroidal triple is called \textit{AT-free}. 

\begin{lemma}
[\cite{ries2009}] \label{l:AT-free} Let $v$ be any vertex of a $B_1$-EPG graph $G$. Then $G[N(v)]$ is AT-free.
\end{lemma}

Let $C$ be any subset of the vertices of a graph $G$. The \textit{branch graph} $B(G|C)$, see~\cite{golumbic2009}, of $G$ over $C$ has a vertex set, $V(B)$, consisting of all the vertices of $G$ not in $C$ but adjacent to some member of $C$, i.e. $V(B) = N(C) - C$. Adjacency in $B(G|C)$ is defined as follows: we join two vertices $x$ and $y$ by an edge in $E(B)$ if and only if in $G$ occurs:
\begin{enumerate}
    \item  $x$ and $y$ are not adjacent;
    \item $x$ and $y$ have a common neighbor $u \in C$;
    \item the sets $N(x) \cap C$ and $N(y) \cap C$ are not comparable, i.e. there exist private neighbors $w, z \in C$ such that $w$ is adjacent to $x$ but not to $y$, and $z$ is adjacent to $y$ but not to $x$; we say that $x$ and $y$ are neighborhood incomparable.
\end{enumerate}

%A graph $G$ is \textit{k-colorable} if its vertices can be colored with at most $k$ colors in such a way that no two adjacent vertices share the same color. 
We let  $\chi(G)$ denote the chromatic number of $G$. 
%\la{donde se usa colorable? creo que acá debe estar el/los resulados que se usan al final de la demo del teo 4.1}

\begin{lemma}[~\cite{golumbic2009}] \label{l:branch} Let $C$ be any maximal clique of a $B_1$-EPG  graph $G$. Then, the branch graph $B(G|C)$ is $\{P_6, \, C_n \hbox{ for }  n\geq 4\}$-free, and $\chi(B(G/C))\leq 3$.
\end{lemma}

\section{Subclasses of Helly-$B_1$-EPG Graphs}

In this section, we delimit some  subclasses of $B_1$-EPG graphs that admit a Helly-$B_1$-EPG representation. It is known that $B_1$-EPG and Helly-$B_1$ EPG 
are hereditary classes, so they can  be characterized by forbidden structures. 
In both cases, finding the list of minimal forbidden induced subgraphs are challenging open problems.
Taking a step towards solving
those problems,  we describe a few structures % that  provide a $B_1$-EPG graph does not admit a Helly-$B_1$ EPG representation, that is it is not a  Helly-$B_1$ EPG graph. 
at least one of which will  necessarily be present in  any $B_1$-EPG graph that does not admit a Helly representation. 
In addition,
we show that the well known families of Block graphs, Cactus and Line of Bipartite graphs are totally contained in the class Helly-$B_1$ EPG.

Let $S_{3}, S_{3'}, S_{3''}$ and $ C_{4}$ be the graphs depicted in Figure \ref{fig:proibidos}.

\begin{theorem}
\label{lem:chordalDiamondFree}
Let $G$ be a $B_1$-EPG graph. If $G$ is  $\{S_{3}, S_{3'}, S_{3''}, C_{4}\}$-free then $G$  is a Helly-$B_1$-EPG graph.
\end{theorem}

\begin{proof}
If $G$ is not a Helly-$B_1$-EPG graph, then in each $B_1$-EPG representation of $G$, there is at least one clique that is represented as claw-clique and no as edge-clique. Consider any $B_1$-EPG  representation of $G$  and let $K$ be a maximal clique  which is represented as a claw-clique. Assume, w.l.o.g,  $K$ is on a claw of the grid with base $[x_0, x_2]\times\{y_0\}$ and center $C = (x_1, y_0)$. Denote by  $\mathcal{P}_K$ the set of paths corresponding to the vertices of $K$.  By Lemma~\ref{lem:lemaBRies2009},  %(see~\cite{ries2009})
%no path $P_w$ for $w\notin K$ covers 
the grid segment $[x_0, x_2]\times\{y_0\}$ is covered only by vertices of $K$. % because $G$ is $C_4$-free

 For every ${\displaystyle \lrcorner}$-path %$P_v \in \mathcal{P}_K$ 
 (resp. ${\displaystyle \llcorner}$-path 
% $P_{v'} \in \mathcal{P}_K$
 ) belonging to $\mathcal{P}_K$, we do the following: if %$P_v$ (resp. $P_{v'}$)
 the path does not intersect any path $P_t \notin\mathcal{P}_K$ on column $x_1$, then we delete its vertical segment and add the grid segment $[x_1, x_2]\times\{y_0\}$ (resp. $[x_0, x_1]\times\{y_0\}$). If after this transformation there is no more ${\displaystyle \lrcorner}$-paths (resp. ${\displaystyle \llcorner}$-paths) in $\mathcal{P}_K$, then we are done since we have obtained an edge-clique. So we may assume that
 every ${\displaystyle \lrcorner}$-path   and every ${\displaystyle \llcorner}$-path  in $ \mathcal{P}_K$ intersects some path $P_t \notin \mathcal{P}_K$   on column $x_1$ (notice that we can assume is the same path $P_t$ for all the vertices). 
 
 Now, if none of the ${\displaystyle \lrcorner}$-paths belonging to $\mathcal{P}_K$ intersects  a path non in  $ \mathcal{P}_K$ on the line $y_0$, then we can replace the horizontal part of those paths by the segment $[x_1,x_2]\times \{y_0\}$, getting an edge representation of the clique $K$. Thus, we can assume there exists
 at least one ${\displaystyle \lrcorner}$-path $P_{v} \in \mathcal{P}_K$ intersecting some path  $P_{t'} \notin \mathcal{P}_K$ on line $y_0$. Analogously, there exists
 at least one ${\displaystyle \llcorner}$-path $P_{v'} \in \mathcal{P}_K$ intersecting some path  $P_{t''} \notin K$ on line $y_0$, as depicted in Figure~\ref{fig:clawGrid}. Notice that vertex $t'$ cannot be adjacent to any of the vertices $t$, $v'$ or $t''$; and, in addition, vertex $t''$ cannot
 be adjacent to   $t$,  or $v$.
 
 Finally,   since $K$ is claw-clique,  there is a path $P_u \in \mathcal{P}_K$ covering the base of the claw. Depending on the 
 possibles adjacencies between  $u$ and $t'$ or  $t''$, one of the graphs  $S_{3}$, $S_{3'}$ or $S_{3''}$ is obtained.

\end{proof}

\begin{figure}
  \centering
  \begin{tabular}{  c p{0.7cm} c}
    %\centering
    \includegraphics[width=5.5cm]{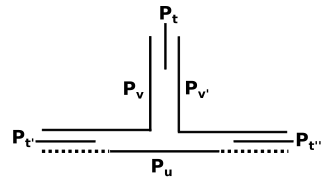} & &
    \includegraphics[width=3.5cm]{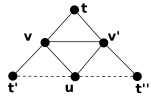}
    \\
    \footnotesize %\centering 
    (a)  \footnotesize Claw with paths. && \footnotesize (b) Subgraph induced by paths.\\
  \end{tabular}

 \caption{Reconstruction of the intersection model.}
 \label{fig:clawGrid}
\end{figure} 

Notice that any bull-free graph is $\{S_{3}, S_{3'}, S_{3''}\}$-free, so our previous result implies  Lemma 5 of  \cite{ries2009}.

\begin{figure}[htb]
  \centering
  \begin{tabular}{  c p{0.7cm} c }
    \centering
    \includegraphics[width=4cm]{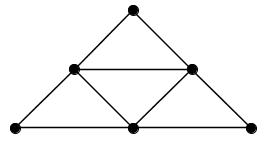} & &
    \includegraphics[width=4cm]{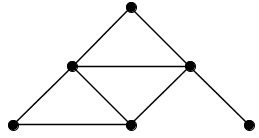}
    \\
    \footnotesize \centering 
    (a)  \footnotesize Graph $S_3$. &&  \footnotesize (b) Graph $S_{3'}$. \\
    
    %---------------------
      \centering 
      \includegraphics[width=4cm]{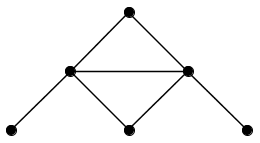} & &
    \includegraphics[width=3cm]{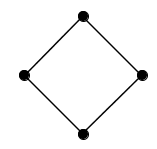}
    \\
    \footnotesize \centering 
    (c)  \footnotesize Graph $S_{3''}$. && \footnotesize (b) Graph $C_{4}$.\\
  \end{tabular}

 \caption{Graphs on the statement of Theorem \ref{lem:chordalDiamondFree}.}
 \label{fig:proibidos}
\end{figure} 

%--------------------------

Next theorem has as consequence the identification of several graph classes where the existence of a $B_1$-EPG representation ensures the existence of a Helly-$B_1$-EPG representation.

\begin{theorem} \label{lem:b1DiamondFree}
 If $G$ is a $B_1$-EPG and diamond-free graph then $G$ is a Helly-$B_1$-EPG graph.
 \end{theorem}

\begin{proof}
If $G$ is not a Helly-$B_1$-EPG graph, then in each $B_1$-EPG representation of $G$, there is at least one clique that is represented as claw-clique and no as edge-clique.  Consider any $B_1$-EPG  representation of $G$  and let $K$ be a maximal clique  which is represented as a claw-clique. Assume, w.l.o.g,  $K$ is on a claw of the grid with base $[x_0, x_2]\times\{y_0\}$ and center $C = (x_1, y_0)$. Denote by  $\mathcal{P}_K$ the set of paths corresponding to the vertices of $K$. 
 By Lemma~\ref{lem:b1epgDiamondFree},  %(see~\cite{ries2009})
%no path $P_w$ for $w\notin K$ covers 
the grid segment $[x_0, x_2]\times\{y_0\}$ is covered only by vertices of $K$. % because $G$ is $C_4$-free
 For every ${\displaystyle \lrcorner}$-path %$P_v \in \mathcal{P}_K$ 
 (resp. ${\displaystyle \llcorner}$-path 
% $P_{v'} \in \mathcal{P}_K$
 ) belonging to $\mathcal{P}_K$, we do the following: if %$P_v$ (resp. $P_{v'}$)
 the path does not intersect any path $P_t \notin\mathcal{P}_K$ on column $x_1$, then we delete its vertical segment and add the grid segment $[x_1, x_2]\times\{y_0\}$ (resp. $[x_0, x_1]\times\{y_0\}$). If after this transformation there is no more ${\displaystyle \lrcorner}$-paths (resp. ${\displaystyle \llcorner}$-paths) in $\mathcal{P}_K$, then we are done since we have obtained an edge-clique. So we may assume that
 every ${\displaystyle \lrcorner}$-path   and every ${\displaystyle \llcorner}$-path  in $ \mathcal{P}_K$ intersects some path $P_t \notin \mathcal{P}_K$   on column $x_1$ (notice that we can assume is the same path $P_t$ for all the vertices). Since  $K$ is claw-clique,  there is a path $P_u \in \mathcal{P}_K$ covering the base of the claw. Thus, $G[v, v', u, t]$ induces a diamond,  a contradiction. 
\end{proof}  

An \textit{independent set} of vertices is a set of vertices no two of which are adjacent.
A graph $G$ is said to be \textit{Bipartite} if its set of vertices can be partitioned into two distinct independent sets.
 There are Bipartite graphs that are non $B_1$-EPG, for instance $K_{2,5}$ and $K_{3,3}$ (see~\cite{cohen2014}). Clearly , since
 bipartite graphs are triangle-free, any $B_1$-EPG representation of a bipartite graph is also a Helly-$B_1$-EPG representation.
 A similar result (but a bit weaker) is obtained as corollary of the previous theorem.

\begin{cor}
If $G$ is a Bipartite $B_1$-EPG graph then $G$ is a Helly-$B_1$-EPG graph.
\end{cor}

\begin{proof}
The Bipartite graphs are diamond-free, thus by Theorem~\ref{lem:b1DiamondFree} these graphs are Helly-$B_1$-EPG graphs.
\end{proof}

A \textit{Block graph} or \textit{Clique Tree} is a type of graph in which every biconnected component (block) is a clique.

\begin{cor}\label{lem:cdf}
 Block graphs are Helly-$B_1$ EPG.
\end{cor}

\begin{proof}
Block graphs are known to be exactly the Chordal diamond-free graphs, so by   Theorem 19 of \cite{ries2009}, all Block graphs are  $B_1$-EPG. If follows from Theorem~\ref{lem:b1DiamondFree} that all Block graphs are Helly-$B_1$ EPG. 
 \end{proof} 

A \textit{Cactus} (sometimes called a Cactus Tree)  graph is a connected graph in which any two  cycles have at most one vertex in common. Equivalently, it is a connected graph in which every edge belongs to at most one  cycle, or (for nontrivial Cactus) in which every block (maximal subgraph without a cut-vertex) is an edge or a cycle. The family of graphs in which each component is a Cactus is closed under graph minor operations. This graph family may be characterized by a single forbidden minor, the diamond graph.
 
\begin{cor}
Cactus graphs are  Helly-$B_1$ EPG.
\end{cor}
\begin{proof}
In~\cite{cela2019monotonic}, it is proved that every Cactus graph is a monotonic $B_1$-EPG graph 
(there is a $B_1$-EPG representation where all paths are ascending in rows and columns). 
Thus, Cactus graphs are $B_1$-EPG graphs. 

Since Cactus are diamond-free, by Theorem ~\ref{lem:b1DiamondFree}, the proof follows.
\end{proof}

Given a graph $G$, its \textit{Line graph} $L(G)$ is a graph such that each vertex of $L(G)$ represents an edge of $G$ and
  two vertices of $L(G)$ are adjacent if and only if their corresponding edges share a common endpoint (i.e. ``are incident'') in $G$.  
A graph $G$ is a \textit{Line graph of a Bipartite graph} (or simply \textit{Line of Bipartite}) if and only if it
contains no claw, no odd cycle (with more than 3 vertices), and no diamond as induced subgraph, \cite{harary1974line}.

In~\cite{daniel2014b} was proved that every Line graph has a representation with at most 2 bends. We proved in the following corollary that when restricted to the Line of Bipartite we can obtain a representation Helly and one-bended.

\begin{cor}\label{coro:lineOfBipartite}
 Line of Bipartite graphs are Helly-$B_1$ EPG. 
\end{cor}

\begin{proof}
Line of Bipartite graphs were proved to be $B_1$-EPG in~\cite{golumbic2018edge}. Since they are diamond-free, the proof follows from Theorem~\ref{lem:b1DiamondFree}.
 
\end{proof}

The diagram of Figure~\ref{fig:diagram}
illustrates the containment relationship between the graph classes  studied so far in this work. 
We list in Figure~\ref{fig:exemplosDiagram} examples of graphs in each numbered region of the diagram. The numbers of each item below correspond to the regions of the same number in the diagram depicted in Figure~\ref{fig:diagram}.

%This numbers correspond with the respective number item and in some cases we make a brief explanation.

\begin{enumerate}[label=(\arabic*)]
    \item ($B_1$-EPG)  - (Helly-$B_1$-EPG) graphs, depicted in Figure~\ref{fig:exemplosDiagram}(a), graph $E_1$;%1
    
    \item (Line of Bipartite)  - (Cactus) - (Block) - (Bipartite) graphs, depicted in Figure~\ref{fig:exemplosDiagram}(b), graph $E_2$;%2
    \item (Helly-$B_1$ EPG) - (Line of Bipartite) - (Block) - (Cactus) - (Bipartite) graphs, depicted in Figure~\ref{fig:exemplosDiagram}(c), graph $E_3$;%3
    \item (Block) $\cap$ (Line of Bipartite) - (Cactus) - (Bipartite), depicted in Figure~\ref{fig:exemplosDiagram}(d), graph $E_4$;%4
    \item (Block) $\cap$ (Line of Bipartite) $\cap$  (Cactus) - (Bipartite), depicted in Figure~\ref{fig:exemplosDiagram}(e), graph $E_5$;%5
    \item (Cactus) $\cap$ (Line of Bipartite) - (Block) - (Bipartite). This intersection is empty. Let $G$ be a graph that is Cactus and Line of Bipartite then $G$ is $\{$claw, odd cycle, diamond$\}$-free. But $G$ is not a Bipartite graph, then $G$ has odd cycle, %. Thus $G$ has at least one triangle or at least one odd cycle $C_n, n\geq 4$, and $G$ is a connected graph. But given a cycle $C_n, n\geq 4$, if add one vertex any adjacent to this cycle then this induce a claw, 
     absurd with the hypothesis of $G$ is Line of Bipartite;%6
    \item (Bipartite) $\cap$ (Line of Bipartite)  - (Cactus) - (Block) graphs, depicted in Figure~\ref{fig:exemplosDiagram}(f), graph $E_7$;%7
    \item (Bipartite) $\cap$ (Line of Bipartite) $\cap$  (Cactus) - (Block) graphs, depicted in Figure~\ref{fig:exemplosDiagram}(g), graph $E_8$;%8
    \item (Bipartite) $\cap$ (Line of Bipartite) $\cap$  (Cactus) $\cap$ (Block) graphs, depicted in Figure~\ref{fig:exemplosDiagram}(h), graph $E_9$;%9
  \item (Bipartite) $\cap$  (Cactus) $\cap$ (Block) - (Line of Bipartite) graphs, depicted in Figure~\ref{fig:exemplosDiagram}(i), graph $E_{10}$;%10
    \item (Bipartite)  $\cap$  (Cactus) - (Block) -  (Line of Bipartite) graphs, depicted in Figure~\ref{fig:exemplosDiagram}(j), graph $E_{11}$;%11
     \item (Bipartite) $\cap$ (Helly-$B_1$ EPG) - (Cactus) - (Block) -  (Line of Bipartite) graphs, depicted in Figure~\ref{fig:exemplosDiagram}(k), graph $E_{12}$;%12
      \item (Bipartite) - ($B_1$-EPG) graphs, depicted in Figure~\ref{fig:exemplosDiagram}(l), graph $E_{13}$;%13
      \item (Block) - (Bipartite) - (Line of Bipartite)  - (Cactus) graphs, depicted in Figure~\ref{fig:exemplosDiagram}(m), graph $E_{14}$;%14
 
      \item (Block) $\cap$  (Cactus) -  (Line of Bipartite) - (Bipartite) graphs, depicted in Figure~\ref{fig:exemplosDiagram}(n), graph $E_{15}$;%15
      \item (Cactus) - (Block) -  (Line of Bipartite) - (Bipartite) graphs, depicted in Figure~\ref{fig:exemplosDiagram}(o), graph $E_{16}$, the odd cycles $C_{2n+1},n\geq 2$;%16
      \item (Helly EPG) - ($B_1$-EPG)  - (Bipartite) graphs, depicted in Figure~\ref{fig:exemplosDiagram}(p), graph  $E_{17}$;%17
\end{enumerate}

 \begin{figure}[htb]	
 \center%6.3
 \includegraphics[width=8cm]{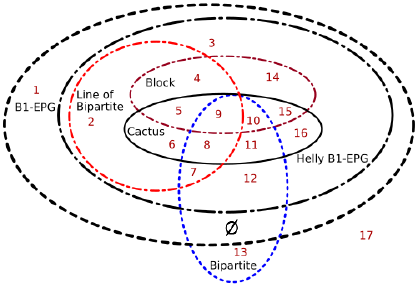}
 \caption{Diagram of some graph classes.}
\label{fig:diagram}
\end{figure}  

 \begin{figure}[htb]	
 
   \centering
  \begin{tabular}{  c c c c }
    %\centering
    \includegraphics[width=3cm]{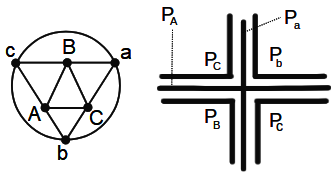} 
    & 
    \includegraphics[width=1.5cm]{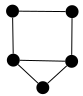} 
    & 
    \includegraphics[width=2cm]{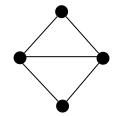} 
    & 
    \includegraphics[width=1.5cm]{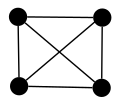} 
    \\
    \footnotesize (a)  Graph $E_1$. 
    & 
    \footnotesize (b) Graph $E_2$.
    & 
    \footnotesize (c) Graph $E_3$.
    & 
    \footnotesize (d) Graph $E_4$.
    \\%%Segunda linha
    \includegraphics[width=2cm]{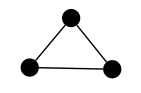} 
    &
    \includegraphics[width=2.5cm]{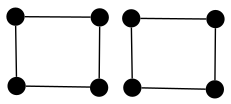} 
    & 
    \includegraphics[width=1.5cm]{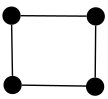} 
    & 
    \includegraphics[width=1.8cm]{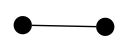} 
    \\ %%Segundo Bloco legendas
    
    \footnotesize (e) Graph $E_5$.
    & \footnotesize (f)   Graph $E_7$. 
    & 
    \footnotesize (g) Graph $E_8$.
    & 
    \footnotesize (h) Graph $E_9$.
    %%Terceira linha de imagens
    \\%%Terceira linha
    \includegraphics[width=1.5cm]{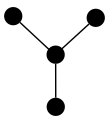} 
    & 
    \includegraphics[width=1.8cm]{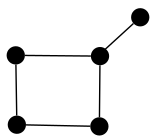} 
    &
    \includegraphics[width=2.5cm]{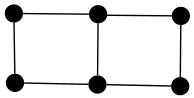} 
    & 
    \includegraphics[width=2cm]{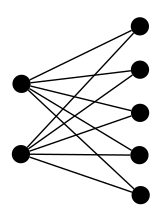} 
    \\ %%Terceiro Bloco legendas
    \footnotesize (i) Graph $E_{10}$.
    & 
    \footnotesize (j) Graph $E_{11}$.
    &
    \footnotesize (k)  Graph $E_{12}$. 
    & 
    \footnotesize (l) Graph $E_{13}$.
    \\
   \includegraphics[width=2cm]{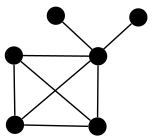} 
    & 
    \includegraphics[width=1.8cm]{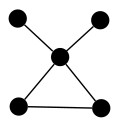} 
    & 
    \includegraphics[width=1.8cm]{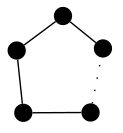} 
    &
    \includegraphics[width=2.5cm]{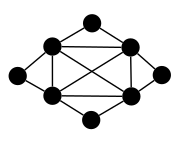}
    \\
    \footnotesize (m) Graph  $E_{14}$.
    & 
    \footnotesize (n) Graph $E_{15}$.
    & 
    \footnotesize (o)  Graph $E_{16}$,  $C_{2n+1},n\geq2$.
    &
    \footnotesize (p)  Graph $E_{17}$.
  
  \end{tabular}
 \caption{The set of instances for the Venn Diagram %of the graph classes of this paper.
 on  Figure~\ref{fig:diagram}.}
 %, see  more in~\cite{leveque2009characterizing,tondato2009grafos}
 \label{fig:exemplosDiagram}
\end{figure}

In next section we explore the Chordal $B_1$-EPG graphs through of a subset of forbidden graphs and we will proof that this class is in the strict intersection of VPT and EPT graphs.

\section{Containment relationship among Chordal $B_1$-EPG, VPT and EPT graphs }

 Any graph that
admits a $B_1$-EPG representation  whose paths do not cover all the edges of a polygon of the grid (i.e.
the subjacent grid subgraph is a tree)  is also an EPT graph: the same representation is both $B_1$-EPG and $EPT$.
However, it is easily verifiable that the subjacent grid subgraph of any $B_1$-EPG representation of a cycle $C_n$ with $n\geq 5$ is not a tree,
%has a non chordal subjacent grid subgraph 
although $C_n$ is an  EPT graph.  Our long-rage goal is 
understanding the $B_1$-EPG graphs that are also EPT graphs. When can a $B_1$-EPG representation
be reorganized into an EPT representation?  In this section,
 we answer that question for Chordal $B_1$-EPG graphs, in fact we prove that every Chordal $B_1$-EPG graph is EPT. We
 made several unsuccessful attempts to prove this result by considering for a graph $G$, a $B_1$-EPG representation whose paths cover all the edges
 of some polygon on the grid, and trying  to show  that if none of the paths could be modified in order to avoid an edge of the polygon,
 then $G$ had some chordless  cycle (i.e. $G$ is not chordal). The surprise was that the only way we found to demonstrate our main Theorem \ref{teo:b1epgept} was through $VPT$ graphs.
 We will prove the following theorem.

\begin{theorem}\label{teo:chordalB1inVPT}
Chordal $B_1$-EPG $\subsetneq$ VPT. 
\end{theorem}

In~L{\'e}v{\^e}que et al. \cite{leveque2009characterizing} apud \cite{alcon2015characterizing},  VPT graphs were characterized by a family of minimal forbidden induced subgraphs,
the ones depicted in 
Figure~\ref{fig:16proibidos} plus the induced cycles $C_n$ for $n\geq 4$. Therefore, in order to prove
that Chordal $B_1$-EPG graphs are VPT is enough to show that none of the graphs in Figure~\ref{fig:16proibidos} 
is $B_1$-EPG. %The following lemmas are developed with that objective.   

First notice that in each one of the graphs $F_{1}, F_{2}, F_{3}, F_{4}$ and $F_{5}$ ( Figures~\ref{fig:16proibidos}(a), (b), (c), (d), (e), respectively), the neighborhood of the universal vertex (the one that is a bit bigger than the others, in the respective figures) contains an asteroidal triple. Therefore, by Lemma \ref{l:AT-free}, these graphs are not  $B_1$-EPG.

Now, in each one of the graphs $F_{11}, F_{12}, F_{13}, F_{14}$, $F_{15}$ and $F_{16}$  (Figures~\ref{fig:16proibidos}(k), (l), (m), (n), (o), (p), respectively), let $C$ be the maximal clique in bold. It is easy to check that, in all cases, the branch graph $B(G|C)$ contains an induced cycle $C_n$, for some $n\geq 4$, or an induced path $P_6$; thus, by Lemma \ref{l:branch},  graphs $F_{11}, F_{12}, F_{13}, F_{14}$, $F_{15}$ and $F_{16}$ are not $B_1$-EPG.

 % A \textit{satellite} of a clique $K$ is a vertex $v$ such that $B_v=N(v)\cap K$ is a 
%nonempty proper subset of $K$. The set $B_v$ is called the \textit{base} of $v$ and it is said \textit{minimal} if no other
%base of a satellite of $K$ is properly contained in $B_v$, see~\cite{alcon2010necessary}.

 %Let $I=[q_1,q_2]$ be the grid interval defined by the intersection $\displaystyle \cap_{v\in K}P_v$, where $K$
%is an edge-clique of a graph $G$. For any $v\in K$, by removing the interval $(q_1,q_2)$, the path $P_v$
%is split into two \textit{disjoint parts}: \textit{part 1}  containing $q_1$, and  \textit{part 2}  containing $q_2$.
%If $w$  is a satellite of $K$ adjacent to $v$, then
%$P_w\cap P_v$ is contained either in part 1 or in part 2 of $P_v$. We will say that $P_w$ intersects $P_v$
%on side 1 or on side 2 of $K$, respectively. Notice that if $w$  is also adjacent
%to another vertex $v'$ of $K$, then   $P_w$ intersects $P_v$ and $P_{v'}$ on
%a same side of $K$. It allow us to divide the satellites of $K$ into two \textit{disjoint
%groups}, the ones on  \textit{side 1} of $K$ and the ones on \textit{side 2}.

%%%%%%%%%%%%%%%%%%%%%%%%%%%
\begin{obs} \label{f:between}Let $e_{\ell}$, $e_m$ and $e_r$  be three distinct edges of a  one-bend path $P$, and assume that $e_m$ is between $e_{\ell}$ and $e_r$ on $P$. If $P_{\ell}$ and $P_r$ are one-bend paths such that: $P_{\ell}$ contains $e_{\ell}$, $P_r$ contains $e_r$, and  $P_{\ell}$ and $P_r$ intersect in at least one edge, then $P_{\ell}$ or $P_r$ contains $e_m$.
\end{obs}
%%%%%%%%%%%%%%%%%%%%%%%%%%%%%%
\begin{obs} \label{f:two points} Let  $e$ and $q$  be an edge and a  point  of a  one-bend path $P$, respectively. If a one-bend path $P'$ contains both $e$ and $q$, then $P'$ contains the whole segment of $P$ between $q$ and $e$.
\end{obs}

\begin{lemma}\label{l:abclique}
Let $G$ be a graph whose vertex set  can be partitioned into a clique $K=\{a,b\}$ and an independent set $I=\{x,y,z\}$, such that each vertex of $K$ is adjacent to each vertex of $I$.
If in a given $B_1$-EPG representation of $G$, $P_a\cap P_y$ is between $P_a\cap P_x$ and $P_a\cap P_z$, then $\{a,b,y\}$ is an edge-clique, and
$P_a\cap P_y \subset P_b$. Even more, any vertex adjacent to both $a$ and $y$, but not to $b$ (or to $b$ and $y$, but not to $a$) has to be adjacent to $x$ or to $z$.
\end{lemma}

\begin{proof}
%%%%%%%%%%%%%%%%%%%%%%%%%%%%%%%
Assume in order to obtain a contradiction that $\{a,b,y\}$ is not an edge-clique. Then, by Lemma
\ref{lem:3cliquesNotClaw}, we can assume, w.l.o.g., that $\{a,b,x\}$ is an edge-clique.
It implies that there is an edge $e_{\ell}$ of $P_a \cap P_x$ covered by $P_b$. Since every  edge
of $P_a\cap P_z$ is covered  by $P_z$, $z$ and $b$ are adjacent, and $z$ and $y$ are non adjacent, we have by Observation \ref{f:between},
that every edge of $P_a\cap P_y$ is covered by $P_b$, which implies  that $\{a,b,y\}$ is an edge-clique, contrary to the assumption.

Thus, $\{a,b,y\}$ is an edge-clique. By Observation \ref{f:two points}, we have that the whole interval of $P_a$ between 
 $P_a \cap P_x$ and   $P_a \cap P_z$ is contained in $P_b$, and so, in particular, $P_a\cap P_y \subset P_b$. Observe that this
 implies that if $q$ is an end vertex of the interval $P_a \cap P_y$, and $e$ is the edge of $P_a$ incident on $q$ that do not belong to $P_y$, then 
$e$ belongs to $P_b$ or to $P_x$ or to $P_z$.
 
 Now, assume there exists a vertex $v$ adjacent to both $a$ and $y$, but not to $b$. Then, the clique $\{a,y,v\}$ has to be a claw-clique. Let $q$ be the center of the claw, notice that $q$ has to be an end vertex of the interval $P_a \cap P_y$.
 Since $v$ is not adjacent to $b$, it follows from the observation at the end of the paragraph above, that
 $v$ has to be adjacent to $x$ or to $z$.

\end{proof}

\begin{lemma}\label{lem:F_6}
The graph $F_6$ on Figure~\ref{fig:16proibidos}(f) is not   $B_1$-EPG.
\end{lemma}
\begin{proof} Let $K=\{1,2\}$ and $I=\{3,4,5\}$. If there exists a $B_1$-EPG representation of $F_6$,  by Lemma \ref{l:abclique},  because of the existence of the vertices $6$, $7$ and $8$, none of the vertices $3$, $4$ and $5$
may intersect $1$ between the remaining two, thus such a representation does not exist.
 \end{proof}

\begin{lemma}\label{lem:F_7}
The graph $F_7$ on Figure~\ref{fig:16proibidos}(g) is not   $B_1$-EPG.
\end{lemma}
\begin{proof} Let $K=\{1,2\}$ and $I=\{4,5,6\}$. If there exists a $B_1$-EPG representation of $F_7$,  by Lemma \ref{l:abclique},  because of the existence of the vertices  $7$ and $8$, the vertex $6$ must intersect  vertex $1$ between $3$ and $4$. But considering $K'=\{1,3\}$, because of the existence of the vertices $5$ and $6$,  vertex $4$ must intersect vertex $1$ between $5$ and $6$. This contradiction implies that such a representation does not exist.
 \end{proof} 
 
 \begin{lemma}\label{lem:F_8_9_10(8)}
The graphs $F_8$, $F_9$ and $F_{10}(8)$ on Figures~\ref{fig:16proibidos}(h), (i) and (j), respectively, are not   $B_1$-EPG.
\end{lemma}
\begin{proof} Let $K=\{2,3\}$ and $I=\{1,6,7\}$. If there exists a $B_1$-EPG representation of any one of those graphs,  by Lemma \ref{l:abclique},  because of the existence of the vertices  $4$ and $5$, the vertex $1$ must intersect  vertex $2$ between $6$ and $7$. In addition, since $\{2,6,8\}$ is a clique, $8$ intersects $2$ in an edge of $P_6\cap P_2$ (edge-clique) or in an edge incident to $P_6\cap P_2$ (claw-clique). Analogously, because of the clique $\{2,7,8\}$,  $8$ intersects $2$ in an edge of $P_7\cap P_2$ (edge-clique) or in an edge incident to $P_7\cap P_2$ (claw-clique). In any case, it implies that $8$ intersects $2$ on two different edges, each one in a different side of $P_2 \cap P_1$, thus, by Observation \ref{f:two points}, $P_8$ contains the interval  $P_2 \cap P_1$, in contradiction with the fact that $1$ and $8$ are not adjacent.
 \end{proof} 
  
  \begin{lemma}\label{lem:F_10(n)}
The graphs  $F_{10}(n)$ for $n \geq 8$ on Figure~\ref{fig:16proibidos}(j) are not   $B_1$-EPG.
\end{lemma}
\begin{proof} The case $n=8$ was considered  in the previous Lemma \ref{lem:F_8_9_10(8)}, so assume $n\geq 9$.   Let $K=\{2,3\}$ and $I=\{1,6,7\}$. If there exists a $B_1$-EPG representation of any one of those graphs,  by Lemma \ref{l:abclique},  because of the existence of the vertices  $4$ and $5$, the vertex $1$ must intersect  vertex $2$ between $6$ and $7$. In addition, since $\{2,6,8\}$ is a clique, $8$ intersects $2$ in an edge of $P_6\cap P_2$ (edge-clique) or in an edge incident to $P_6\cap P_2$ (claw-clique). Analogously, because of the clique $\{2,7,n\}$,  $n$ intersects $2$ in an edge of $P_7\cap P_2$ (edge-clique) or in an edge incident to $P_7\cap P_2$ (claw-clique). In any case, it implies that $8$  and $n$ intersect $2$ on two different edges, each one in a different side of $P_2 \cap P_1$. Therefore, there exist two consecutive vertices of the path $8, 9, \ldots, n$, say the vertices $j$ and $j+1$, such that each one intersects $P_2$ on a different side of $P_2 \cap P_1$. Thus, by Observation \ref{f:between}, $P_j$ or $P_{j+1}$ must contain the interval  $P_2 \cap P_1$, in contradiction with the fact that neither $j$ nor $j+1$ is adjacent to $1$.
 \end{proof}

\begin{landscape}
 \begin{figure}[htb]	
 
   \centering
  \begin{tabular}{ccccc}
    %\centering
    \includegraphics[width=2.5cm]{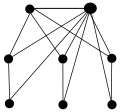} 
    & 
    \includegraphics[width=2.3cm]{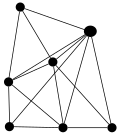} 
    & 
    \includegraphics[width=3cm]{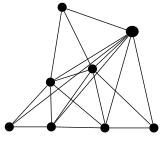} 
    & 
    \includegraphics[width=3cm]{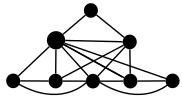} 
    & 
    \includegraphics[width=3cm]{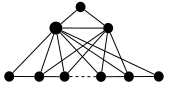} 
    \\
    \footnotesize 
    (a)  \footnotesize Graph $F_1$. 
    & 
    \footnotesize (b) Graph $F_2$.
    & 
    \footnotesize (c) Graph $F_3$.
    & 
    \footnotesize (d) Graph $F_4$.
    & 
    \footnotesize (e) Graph $F_5(n),n\geq7$.
    \\%%Segunda linha
        \includegraphics[width=2.5cm]{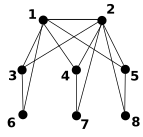} 
    & 
    \includegraphics[width=3.5cm]{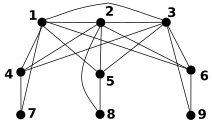} 
    & 
    \includegraphics[width=3cm]{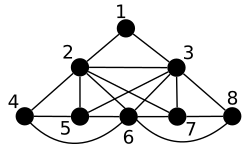} 
    & 
    \includegraphics[width=3cm]{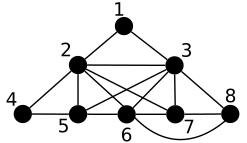} 
    & 
    \includegraphics[width=3cm]{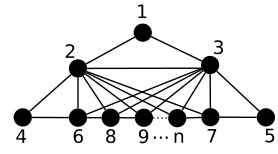} 
    \\ %%Segundo Bloco legendas
    \footnotesize 
    (f)  \footnotesize Graph $F_6$. 
    & 
    \footnotesize (g) Graph $F_7$.
    & 
    \footnotesize (h) Graph $F_8$.
    & 
    \footnotesize (i) Graph $F_9$.
    & 
    \footnotesize (j) Graph $F_{10}(n), n\geq  8$.
    %%Terceira linha de imagens
    \\%%Terceira linha
        \includegraphics[width=3cm]{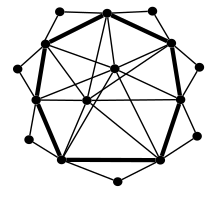} 
    & 
    \includegraphics[width=3cm]{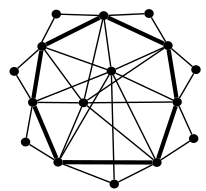} 
    & 
    \includegraphics[width=3cm]{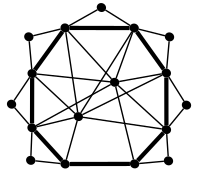} 
    & 
    \includegraphics[width=3cm]{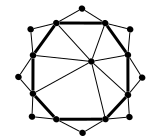} 
    & 
    \includegraphics[width=3cm]{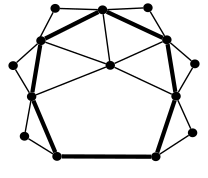} 
    \\ %%Terceiro Bloco legendas
    \footnotesize 
    (k)  \footnotesize  $F_{11}(4k),k\geq2$. 
    & 
    \footnotesize (l)  $F_{12}(4k),k\geq2$.
    & 
    \footnotesize (m)  $F_{13}(4k+1),k\geq2$.
    & 
    \footnotesize (n)  $F_{14}(4k+1),k\geq2$.
    & 
    \footnotesize (o)  $F_{15}(4k+2),k\geq2$.
    
    \\ %Ultima linha Figuras
    
    && \includegraphics[width=3cm]{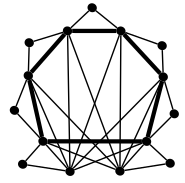} &&
    
    \\%Ultima linha Legendas
    
    && \footnotesize (p)  $F_{16}(4k+3),k\geq2$. &&
    
    %\multicolumn{3}{c}{ \footnotesize (c) Another partial single bend representation of $H$ } \\
  \end{tabular}
 \caption{The 16 Chordal induced subgraphs forbidden to VPT (the vertices in the cycle marked by bold edges form a clique).}
 %, see  more in~\cite{leveque2009characterizing,tondato2009grafos}
 \label{fig:16proibidos}
\end{figure}  
 \end{landscape}

We have proved that every minimal forbidden induced subgraph for VPT is also a  forbidden induced subgraph for Chordal $B_1$-EPG. Moreover, there are graphs in VPT that do not belong to $B_1$-EPG, for instance the graph $4$-sun $S_4$ is not in $B_1$-EPG, see~\cite{golumbic2009}, but it has a VPT representation, see Figures~\ref{fig:exemplos}(a) and~\ref{fig:exemplos}(b). Thus, VPT graphs  properly contain Chordal $B_1$-EPG graphs. This ends the proof of Theorem~\ref{teo:chordalB1inVPT}. 

  \begin{cor}\label{c:prohib}
  Each one of the graphs depicted on Figure \ref{fig:16proibidos} is a forbidden induced subgraph for the class $B_1$-EPG.
  \end{cor}
%  \la{agregar corolario: cada uno de los grafos en el la figura es un prohibido minimal para $B_1$-EPG}
  
 % \tanilson{recordar que o grafo de 9(j) nao eh minimal, uma vez que agregamos alguns vertices.}

%\la{Continuar aqui. En lo que sigue hay varias cosas que no han sido definidas. Ver si conviene definirlas o expresarlas de otra manera} 

%\input{./includes/include-img/exemplos.tex}
\begin{figure}[htb]
  \centering
  \begin{tabular}{ c c c }
    \centering
    \includegraphics[width=5cm]{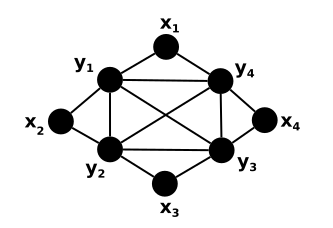} & &
    \includegraphics[width=8cm]{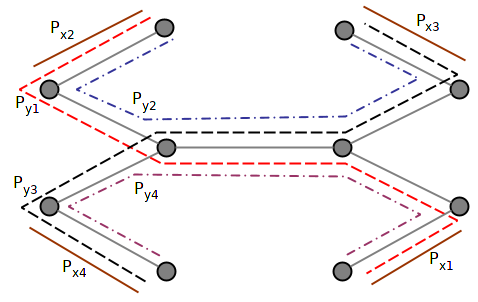}
    \\
    \footnotesize \centering 
    (a)  \footnotesize Graph $S_4$. &&  \footnotesize (b) A VPT and EPT representation of $S_{4}$. \\

  \end{tabular}

 \caption{Graph $S_4$ and one of its possible VPT and EPT representations.}
 \label{fig:exemplos}
\end{figure}

%Grafo C_4 eh EPT mas nao eh Chordal
\begin{theorem}\label{teo:b1epgept}
Chordal $B_1$-EPG $\subsetneq$ EPT. 
\end{theorem}

\begin{proof}
Let $G$ be a  Chordal $B_1$-EPG graph. By the previous Theorem \ref{teo:chordalB1inVPT}, $G$ is VPT. And,  by Lemma \ref{l:branch}, $\chi(B(G/C))\leq 3$ for every maximal clique $C$ of $G$.   In \cite{alcon2014recognizing} (see Theorem 10), it was proved that if the chromatic number of the branch graph of a VPT graph is at most $h$ for every maximal clique, then the graph admits a VPT representation on a host tree with maximum degree $h$.  Therefore, $G$ admits a VPT representation on a host tree with maximum degree 3.  Finally, in \cite{ golumbic1985edge} (see Theorem 2), it was prove that any VPT graph that  admits a representation on a host tree with maximum degree 3 is also an EPT graph. Consequently,  $G$ is EPT.

The same graph   $S_4$ used in the proof of the previous theorem (see Figure~\ref{fig:exemplos}(b)) shows that there are EPT graphs that are not $B_1$-EPG. 
%
%and let  $\mathcal{C}$  be the set of cliques of $G$. By Theorem~\ref{teo:chordalB1inVPT}, we know that $G$ is VPT. By \cite{golumbic2009}, we know that if $G$ is $B_1$-EPG then $\chi (B(G|C))\leq 3$,  for all $C \in \mathcal{C}$. By a result given in ~\cite{alcon2014recognizing}, we can say that $G \in [3,2,1]$, and $[3,2,1] = VPT \cap EPT$, and $[3,2,1] = [3,2,2] = EPT$ $\cap$ Chordal~\cite{golumbic1985} then $G$ is a Chordal EPT graph. Therefore, Chordal $B_1$-EPG is contained in EPT. Moreover, the inclusion is proper because for example the graph $4$-sun $S_4$ is an EPT graph (see Figure~\ref{fig:exemplos}(b)) but it is not a $B_1$-EPG graph (see ~\cite{golumbic2009}).
\end{proof}

\section{Conclusion and Open Questions}

In this paper, we have considered three different path-intersection graph classes: $B_1$-EPG, VPT and EPT graphs. We showed that  $\{S_3, S_{3'},S_{3''},C_4\}$-free graphs and others non-trivial subclasses of  $B_1$-EPG graphs have are Helly-$B_1$-EPG, namely by instance Bipartite, Block, Cactus and Line of Bipartite graphs. 
  
 We presented an infinite family of forbidden induced subgraphs for the class  $B_1$-EPG and in particular we proved  that Chordal $B_1$-EPG $\subset$ VPT $\cap$ EPT.

%If on the one hand some few graph classes are known to be properly contained in $B_1$-EPG, for instance the $L$-shaped paths graphs see~\cite{cameron2016edge},  and the recognition time for $B_1$-EPG graphs in general is $NP$-complete. On the other hand, in the course of this section we also present some subclasses of Helly-$B_1$ EPG for which the recognition problem is polynomial.

In~\cite{ries2009}, Asinowski and Ries described the   Split graphs that are $B_1$-EPG graphs in case the the stable set  or the  central  size have size three. 
The graphs $F_2, F_{11}, F_{13}, F_{14}$ and $F_{15}$, given in Figure~\ref{fig:16proibidos} are Split, we have  used a different approach  to prove that they are not $B_1$-EPG graphs. So one question is pertinent: Can we characterize Split graphs in general based in results of this paper? 

Finally, another interesting research would be to explore families of Helly-EPG graphs more deeply. We would like to understand the behavior of other graph classes inside $B_1$-EPG graph class, i.e. if given an  input graph $G$ that is an instance of (for example) Weakly Chordal $B_1$-EPG. What is the relationship of $G$ with the EPT/VPT graph class? What happens when we demand that the representations be Helly-$B_1$ EPG? Does  recognizing problem remains hard for each one of these classes?

\section*{Acknowledgement}

The present work was done while the third author was a doctoral research fellow at National University of La Plata - UNLP, Math Department. The support of this institution is gratefully acknowledged.

The third author (Tanilson) would like to thank the partial financing of this study by the Coordena{\c c}\~ao de Aperfei{\c c}oamento de Pessoal de N\'ivel Superior - Brasil (CAPES) - Finance Code 001.


\begin{thebibliography}{99}

\bibitem{alcon2014recognizing}
L. Alc{\'o}n, M. Gutierrez and M.P. Mazzoleni. \emph{Recognizing vertex intersection graphs of paths on bounded degree trees}. Discrete Applied Mathematics, 162 (2014), 70-77.
% Alc{\'o}n, Liliana and Gutierrez, Marisa and Mazzoleni, Mar{\'\i}a P{\'\i}a. Recognizing vertex intersection graphs of paths on bounded degree trees. Discrete Applied Mathematics, 162 (2014), 70-77.
 %1 old --> new 11

\bibitem{alcon2015characterizing}
  L. Alc{\'o}n, M. Gutierrez,  and M.P. Mazzoleni. \emph{Characterizing paths graphs on bounded degree trees by minimal forbidden induced subgraphs}. Discrete Mathematics, 338 (2015), 103-110. 
%2 old --> new 12

\bibitem{ries2009}
A. Asinowski and B. Ries. \emph{Some properties of edge intersection graphs of single bend paths on a grid}.  Electronic Notes in Discrete Mathematics, 312 (2012), pp. 427-440.
%3 old --> new 1

\bibitem{Asinowski2009}
A. Asinowski and A. Suk. \emph{Edge intersection graphs of systems of paths on a grid with a bounded 	number of bends}. Discrete Applied Mathematics, 157 (2009), pp. 3174-3180.
%4 old --> new 2

\bibitem{dmtcs:6506} %bornstein2019complexity
  C.F. Bornstein, M.C. Golumbic, T.D. Santos, U.S. Souza,  and J.L. Szwarcfiter. \emph{The Complexity of Helly-$B_{1}$ EPG Graph Recognition}. Discrete Mathematics \& Theoretical Computer Science, 22 (2020), https://dmtcs.episciences.org/6506/pdf.
%5 old --> new 4

 \bibitem{cela2019monotonic}
E. Cela and E. Gaar. \emph{Monotonic Representations of Outerplanar Graphs as Edge Intersection Graphs of Paths on a Grid}. arXiv preprint arXiv:1908.01981 (2019).
%6 old --> new 6

\bibitem{cohen2014}
E. Cohen, and M. C. Golumbic and B. Ries. \emph{Characterizations of cographs as intersection graphs of paths on a grid}. Discrete Applied Mathematics, 178 (2014), pp. 46-57.
%7 old --> new 7

\bibitem{gavril1974intersection}
 F. Gavril. \emph{The intersection graphs of subtrees in trees are exactly the chordal graphs}.  Journal of Combinatorial Theory, Series B, 16 (1974), pp. 47-56.
%8 old --> new 8

\bibitem{gavril1978recognition}
F. Gavril. \emph{A recognition algorithm for the intersection graphs of paths in trees}. Discrete Mathematics, 23 (1978), pp. 211-227.
%9 old --> new 9

\bibitem{golumbic1985edge}
M.C. Golumbic and R.E. Jamison. \emph{Edge and vertex intersection of paths in a tree}. Discrete Mathematics, 55 (1985), pp. 151-159.
%10 old --> new 13

\bibitem{golumbic1985}
M.C. Golumbic and R.E. Jamison. \emph{The edge intersection graphs of paths in a tree}. Journal of Combinatorial Theory, B 38 (1985), pp. 8-22.
%11 old --> new 14

\bibitem{golumbic2009}
M.C. Golumbic, M. Lipshteyn and  M. Stern. \emph{Edge intersection graphs of single bend paths on a grid}. Networks, 54 (2009), pp. 130-138. 
%12 old --> new 15

\bibitem{golumbic2013}
M.C. Golumbic, M. Lipshteyn and M. Stern. \emph{Single bend paths on a grid have strong {H}elly number 4}. Networks, 62 (2013), pp. 161-163.
%13 old --> new 16

\bibitem{golumbic2018edge}
M.C. Golumbic, G. Morgenstern and D. Rajendraprasad. \emph{Edge-intersection graphs of boundary-generated paths in a grid}. Discrete Applied Mathematics, 236 (2018), pp. 214-222.
%14 old --> new 17

\bibitem{golumbic2013intersection}
M.C. Golumbic and B. Ries. \emph{On the intersection graphs of orthogonal line segments in the plane: characterizations of some subclasses of chordal graphs}. Graphs and Combinatorics, 29 (2013), pp. 499-517.
%15 old --> new 18

\bibitem{harary1974line}
F. Harary and C. Holzmann. \emph{Line graphs of bipartite graphs}. Revista de La Sociedad Matematica de Chile, 1 (1974), pp. 19-22.
%16 old --> new 10

  \bibitem{daniel2014b}
D. Heldt, K. Knauer and T. Ueckerdt. \emph{On the bend-number of planar and outerplanar graphs}. Discrete Applied Mathematics, 179 (2014), pp. 109-119.
  %17 old --> new 5
  
\bibitem{leveque2009characterizing}
  B. L{\'e}v{\^e}que, F. Maffray  and M. Preissmann. \emph{Characterizing path graphs by forbidden induced subgraphs}. Journal of Graph Theory, 62 (2009), pp. 369-384.
  %18 old --> new 3


\bibitem{syslo1985triangulated}
  M.M. Sys{\l}o. \emph{Triangulated edge intersection graphs of paths in a tree}. Discrete mathematics,  55 (1985)  pp. 217-220.
%19 old --> new 19
\end{thebibliography}
\end{document}